\documentclass[conference]{IEEEtran}
\IEEEoverridecommandlockouts
\usepackage{cite}
\usepackage{amsmath,amstext,amssymb,amsopn,amsthm}
\usepackage{mathtools}
\usepackage{bbm}
\usepackage{amsmath,amssymb,amsfonts}
\usepackage{algorithmic}
\usepackage{subfig}
\usepackage{float}
\usepackage{graphicx}
\usepackage{textcomp}

\newtheorem{thm}{Theorem}

\newtheorem{example}{Example}
\newtheorem{lem}{Lemma}
\usepackage{xcolor}
\def\BibTeX{{\rm B\kern-.05em{\sc i\kern-.025em b}\kern-.08em
    T\kern-.1667em\lower.7ex\hbox{E}\kern-.125emX}}
\begin{document}

\title{A Hybrid Scattering Transform for Signals with Isolated Singularities\\
\thanks{J.H. was supported by the National Science Foundation (grant \#1845856). M.I. was supported in part by NSF DMS 1912706. M.H. acknowledges funding from the National Institutes of Health (grant \#R01GM135929), the National Science Foundation (CAREER award \#1845856), and the Department of Energy (grant \#DE-SC0021152).}
}

\author{\IEEEauthorblockN{Michael Perlmutter}
\IEEEauthorblockA{\textit{Dept. of Mathematics} \\
\textit{University of California, Los Angeles}\\
Los Angeles, CA, USA \\
perlmutter@math.ucla.edu}
\and
\IEEEauthorblockN{Jieqian He}
\IEEEauthorblockA{\textit{Dept. of CMSE} \\
\textit{Michigan State University}\\
East Lansing, MI, USA \\
hejieqia@msu.edu}
\and
\IEEEauthorblockN{Mark Iwen}
\IEEEauthorblockA{\textit{Dept. of Mathematics} \\
\textit{Dept. of CMSE} \\
\textit{Michigan State University}\\
East Lansing, MI, USA \\
iwenmark@msu.edu}
\and

\IEEEauthorblockN{Matthew Hirn}
\IEEEauthorblockA{\textit{Dept. of CMSE} \\
\textit{Dept. of Mathematics}
 \\
\textit{Michigan State University}\\
East Lansing, MI, USA \\
mhirn@msu.edu}
}

\maketitle

\begin{abstract}
The scattering transform is a wavelet-based model of Convolutional Neural Networks originally introduced by S. Mallat. Mallat's analysis shows that this network has desirable stability and invariance guarantees and therefore helps explain the observation that the filters learned by early layers of a Convolutional Neural Network typically resemble wavelets. Our aim is to understand what sort of filters should be used in the later layers of the network. Towards this end,  we propose a two-layer hybrid scattering transform. In our first layer, we convolve the input signal with a wavelet filter transform to promote sparsity, and, in the second layer, we convolve with a Gabor filter to leverage the sparsity created by the first layer. We show that these measurements characterize information about signals with isolated singularities. We also show that the Gabor measurements used in the second layer can be used to synthesize sparse signals such as those produced by the first layer.
\end{abstract}


\section{Introduction}\label{sec: intro}

The wavelet scattering transform is a mathematical model of Convolutional Neural Networks (CNNs)  introduced by  S. Mallat\cite{mallat:scattering2012}. Analogously to the feed-forward portion of a CNN, it produces a latent representation of an input signal via  an alternating sequence of filter convolutions and nonlinearities. It differs, 
most notably, by using predesigned wavelet filters rather than filters learned from data.

Using predefined filters allows for rigorous  analysis and helps us understand why  a deep nonlinear network is better than a wide, shallow, linear  network with the same number of parameters. Ideally, a feed-forward  network should produce a representation which is sufficiently descriptive for downstream tasks, but also stable to deformations such as translations. Linear networks are typically unable to do both and often must discard high-frequency information to achieve stability. Mallat's analysis in \cite{mallat:scattering2012} shows that the scattering transform, on the other hand, captures high-frequency information via  wavelets and then pushes it down to lower, more stable, frequencies using a nonlinear activation function. Thus, the nonlinear structure enables the network to stably capture high-frequency information.

 The scattering transform also helps us understand which filters are useful for effectively encoding information. While the optimal choice is task dependent, wavelets are  often a good choice since natural images are typically sparse in the wavelet basis and as discussed above, they are able to capture high-frequency information.  Moreover, and perhaps most importantly, the filters learned in the early layers of CNNs typically resemble wavelets. 

This paper focuses on the choice of filters for later layers of the network. In particular,
we propose a two-layer hybrid scattering model. In the first layer, we  use a wavelet convolution to sparsify the input. Then, we use a Gabor type filter to leverage this sparsity.

\begin{figure}
    \centering
\includegraphics[scale=0.33]{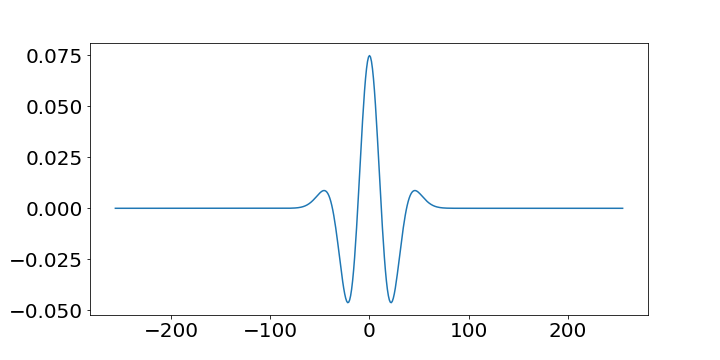}\caption{Wavelet filter used in the first layer}
\end{figure}

 For simplicity, we assume that the input $y(t)$ is a piecewise polynomial whose knots are located at points $\{u_i\}_{i=1}^k$ $u_i<u_{i+1}$. %are located on the grid \textcolor{red}{Why do we need the grid? Should this change the statement of theorem 1?} $h\mathbb{Z}=\{hn:n\in\mathbb{Z}\}$ for some $h>0$. 
 We shall also assume that each of its piecewise components $y_i(t)$  has degree at most $m$. We  let $\psi$ be a mother wavelet with  $\text{supp}(\psi)\subseteq [-1,1]$ and let 
\begin{equation*}
    \psi_\ell(t)=\frac{1}{2^\ell}\psi\left(\frac{t}{2^\ell}\right).
\end{equation*}
We will assume that $\psi$ has $m+1$ vanishing moments, which implies that $\psi_\ell\star y_i(t)=0$ (see e.g. \cite{Mallat:2008:WTS:1525499}).
It follows that    $\text{supp}(\psi_\ell \star y)$ is contained in $\cup_{i=1}^k[u_i-2^\ell,u_i+2^\ell]$.  To further promote sparsity, we next apply a max-pooling operator:
\begin{equation*}
    MP_\ell z(t) = \begin{cases}
    z(t) &\text{if } z(t) =\max_{t'\in [t_i-2^\ell,t_i+2^\ell]\cap h\mathbb{Z}}z(t')\\
    0 &\text{otherwise}
    \end{cases}.
\end{equation*}
As summarized in the following theorem, this yields a  linear combination of Dirac delta functions. 

\begin{figure}
    \centering
    \subfloat[piecewise polynomial $y(t)$]{\label{fig:syn_directed}\includegraphics[scale=0.2]{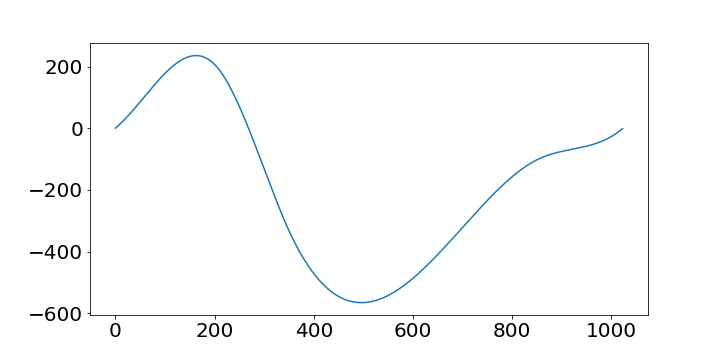}}
       \hfill \subfloat[Convolution against wavelet $\psi_\ell$]{\label{fig:syn_directed}\includegraphics[scale=0.2]{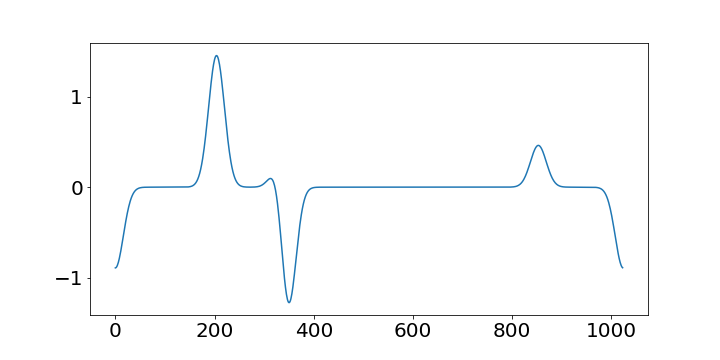}}\hfill
    \subfloat[Sparse Signal from Max Pooling]{\label{fig:syn_sym}\includegraphics[scale=0.2]{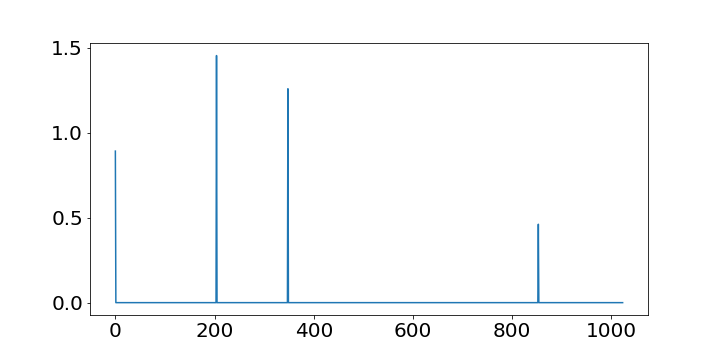}}
    \caption{Wavelets sparsify piecewise polynomials on the interval $[0,1024h]$.}
    \label{fig:diracs}
\end{figure}

\begin{thm}\label{thm: make sparse}
Assume that $2^{\ell+1} \leq \min_{i\neq i'}|u_i-u_{i'}|$. Then,  
\begin{equation*}
    MP_\ell(|\psi_\ell\star y|)(t)=\sum_{j=1}^k a_j \delta_{v_j}(t).
\end{equation*}
for some $a_1,\ldots,a_k> 0,$   $v_j\in [u_j-2^\ell,u_j+2^\ell], 1\leq j\leq k.$
\end{thm}

\begin{figure}
    \centering
    \subfloat[Indicator function window]{\label{fig:syn_sym}\includegraphics[scale=0.32]{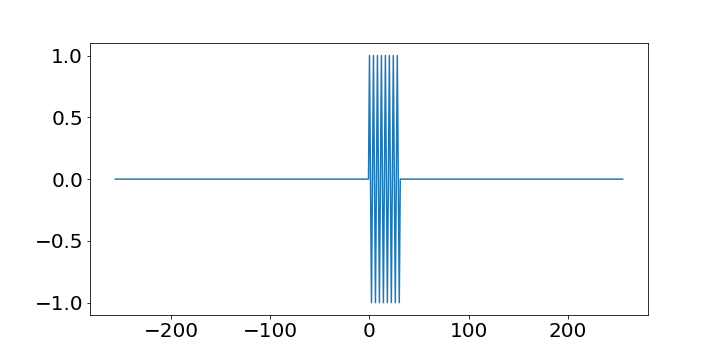}}
    \hfill
    \subfloat[Gaussian window]{\label{fig:syn_sym}\includegraphics[scale=0.32]{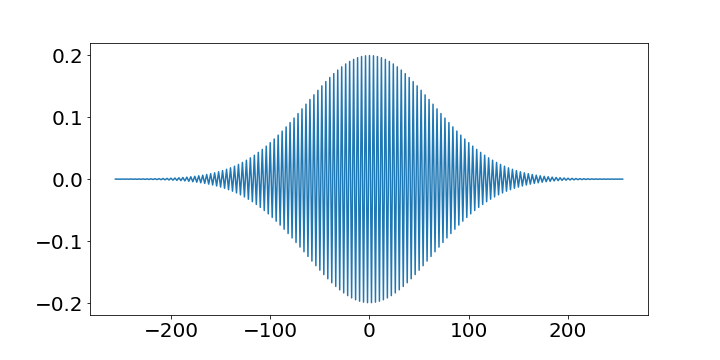}}
    \caption{Gabor filters used in the second layer (real parts)}
    \label{fig:diracs}
\end{figure}

In our second layer, rather than  another wavelet, we  use a Gabor filter 
\begin{equation}\label{eqn: gabor}
  g_{s,\xi} = w\left(\frac{t}{s}\right)\mathbbm{e}^{\mathbbm{i}\xi t},
\end{equation} where the parameters $s$ and $\xi$ determine the scale and central frequency and the window function $w$ is supported on an interval of unit length. Next, we take the $L^p$ norm for some integer $p\geq 1.$ As a result, we obtain translation invariant  \textit{hybrid scattering coefficients}
\begin{equation*}
    \|g_{s,\xi}\star MP_\ell(|\psi_\ell\star y|)\|_p.
\end{equation*}

By design, these measurements are invariant to translations, reflections, and global sign changes. We aim to investigate the ability of our measurements to characterize $y$ up to  these natural ambiguities. 
 The wavelet-modulus is known to be a powerful signal descriptor\cite{mallat:waveletPhaseRetrieval2015}.
Therefore, in light of Theorem \ref{thm: make sparse}, we shall analyze the ability of the  measurements   
\begin{equation}\label{eqn: justgabor}
  f_\xi(s)[x]\coloneqq \|g_{s,\xi}\star x\|_p
    \end{equation}
to characterize signals of the form 
\begin{equation}\label{eqn: diractrain}
x(t)=\sum_{j=1}^k a_j \delta_{v_j}(t).
\end{equation}
For such a signal, we will let $\vec{a}$ be the vector] defined by $\vec{a}=(a_1,\ldots,a_k)$ and let $\|\vec{a}\|_p$ denote its $\ell^p$ norm.

 To supplement our theory, we will  show that the measurements \eqref{eqn: justgabor} can be used to reconstruct a sparse signal of the form \eqref{eqn: diractrain} up to translations, reflections and global sign changes in Section \ref{sec: numerics}.

We will show that our measurements characterize the support set $\{v_j\}_{j=1}^k$. For $i<j$, we let $\Delta_{i,j}=v_j-v_i$ and consider the difference set 
\begin{equation*}
\mathcal{D}(x)\coloneqq\{\Delta_{i,j}: 1\leq  i<j\leq k\}.
\end{equation*} We will assume that $x$ is \textit{collision free}, i.e., that  $\Delta_{i,j}\neq \Delta_{i',j'}$ except for when $(i,j)=(i',j')$ and that $\mathcal{D}(x)$ is contained in a fine grid, $h\mathbb{Z}=\{hn:n\in\mathbb{Z}\}$ for some $h>0$. Under these assumptions, it is known \cite{Bekir2007, DBLP:journals/corr/RanieriCLV13} that the support set $\{v_j\}_{j=1}^k$ is  determined (up to reflection and translation) by $\mathcal{D}(x)$ except for in the case where $k=6$ and the $\{v_j\}_{j=1}^6$ belong to a specific parametric family. (See Theorem 1 of \cite{DBLP:journals/corr/RanieriCLV13} for full details. For the remainder of this work, we will assume that $\{v_j\}_{j=1}^k$  does not belong to this family and therefore the support set $\{v_j\}_{j=1}^k$ is determined by $\mathcal{D}(x).$)
 This motivates the following theorem which shows that the measurements \eqref{eqn: justgabor} uniquely determine $\mathcal{D}(x)$. %with a single central frequency and sufficiently many scales.
\begin{thm}\label{thm: main support}
 Let $p\geq 1$ be an integer and let $w(t)=1_{[0,1]}(t)$.  Then for almost every $\xi$, the function 
\begin{equation*}f_\xi(s)=\|g_{s,\xi}\star x\|_p
\end{equation*}
   is piecewise linear. Morover, the set of its isolated singularities is exactly the support set $\mathcal{D}(x)$.
\end{thm}
%Since, $x(t)$ is collision free we have $|\mathcal{D}(x)|=\mathcal{O}(k^2)$. 
Theorem \ref{thm: piecewiselinear} shows that selecting a single random frequency and enough scales $\{s_j\}$ such that there is one $s_j$ in between each element of $\mathcal{D}(x)$ allows us detect the location of each point of $\mathcal{D}(x)$ by evaluating $f_\xi(s)$ at each of the $s_j$  (up to a precision corresponding to the density of the scales). 
The next result shows that the amplitudes $a_j$ can also be recovered with $\mathcal{O}(p)$ randomly chosen frequencies. Thus, the measurements \eqref{eqn: justgabor} characterize sparse signals up to natural ambiguities. 
\begin{thm}\label{thm: heights}
Let $w(t)=1_{[0,1]}(t)$ and, 
let \begin{equation*}
x(t)=\sum_{j=1}^k a_j \delta_{v_j}(t)
\end{equation*} be a sparse signal of the form \eqref{eqn: diractrain}.
Let $\xi_1,\ldots,\xi_{L}$ be i.i.d. standard normal random variables, where $L$ is assumed to be at least $p+2$ if $p$ is even and at least $4p+2$ if $p$ is odd. Then the following uniqueness result holds almost surely:

Let \begin{equation*}
 \widetilde{x}(t)=\sum_{j=1}^k \widetilde{a}_j \delta_{\widetilde{v}_j}(t).
\end{equation*} Suppose that $\|\vec{a}\|_p=\|\vec{\widetilde{a}} \|_p,$ that $\mathcal{D}(\widetilde{x})=\mathcal{D}(x)$, and that
$$\partial_s^2f_{\xi_\ell}[x](d)=\partial_s^2f_{\xi_\ell}[\widetilde{x}](d)$$ for all $d\in\mathcal{D}(x)$ and all $1\leq \ell \leq 4p+2$.

Then we have that $\vec{\widetilde{a}}=\pm\vec{\widetilde{a}}$, and therefore $\widetilde{x}(t)$ is equivalent to $ x(t)$ up to translation,  reflection, and global sign change.
\end{thm}

\section{Generalized Exponential Polynomials}\label{sec: exppoly}
In this section, we will introduce some notation and state some lemmas that are needed in order to prove Theorems \ref{thm: main support} and \ref{thm: heights}. %In particular, we will introduce a class of functions which we call Generalized Exponential  Polynomials and state several lemmas about the properties of these functions. 
For the proof of  the lemmas in this section, please see section \ref{sec: proof aux}. 

We let $\mathcal{E}$ denote the set of functions that can be written as
\begin{equation}\label{defE}
p(\theta) = \sum_{k=1}^N \alpha_k\mathbbm{e}^{\mathbbm{i}\gamma_k\theta}
\end{equation}
where $N\geq 1,$ $\alpha_k, \gamma_k\in\mathbb{R},$ $\alpha_k\neq 0,$ and $\gamma_1<\gamma_2<\ldots<\gamma_N$. Since the $\gamma_k$ are allowed to be arbitrary (possibly negative or irrational) real numbers,  we call these functions {\it{generalized exponential polynomials}}. For  $p\in\mathcal{E},$  we refer to $\gamma_N$ as the degree of $p$. We let $\mathcal{E}(d)$ refer to the set of all $p\in\mathcal{E}$ with $\text{degree}(p)=d,$ %(Note the degree of $p$ may be irrational or even negative.    %If $\gamma_1\geq 0,$ we say that $p$ has no negative part, and let $\mathcal{E}_0(d)$ denote the subset of $\mathcal{E}$ with no negative part. 
and let $\mathcal{E}_0(d)$ denote the set of $p\in\mathcal{E}(d)$ such that $\gamma_1= 0.$

The following lemma shows that each $p\in\mathcal{E}$ has a unique representation as the sum of exponentials, and that therefore, the degree of $p$ is well defined.%In particular, it will follow that every $p$ in $\mathcal{E}$ e\thetatends to a nonzero entire function on the comple\theta plane, and therefore the set of $\theta\in\mathbb{R}$ such that $p(\theta)=0$ will have measure zero.

\begin{lem}\label{lem: welldefE}
Let $p,q\in\mathcal{E},$ with
\begin{equation*}
p(\theta) = \sum_{k=1}^N \alpha_k\mathbbm{e}^{\mathbbm{i}\gamma_k\theta} \text{ and } q(\theta) = \sum_{k=1}^{N'} \beta_k\mathbbm{e}^{\mathbbm{i}\eta_k\theta}.
\end{equation*}
Then $p=q$ if and only if $N=N'$ and for all $k=1,\ldots,N$ $\alpha_k=\beta_k$ and $\gamma_k=\eta_k.$
\end{lem}

 Lemma \ref{lem: welldefE} implies that if $p\in\mathcal{E}(d_1)$ and $q\in\mathcal{E}(d_2)$, then 
\begin{equation} 
\label{degtimes}pq\in\mathcal{E}(d_1+d_2).
\end{equation}
In particular, if $p\in\mathcal{E}_0(d)$
\begin{equation}\label{square}
|p|^2=p\bar{p}\in\mathcal{E}(d+0)=\mathcal{E}(d).
\end{equation}
Furthermore, if $d_2\leq d_1$ then \begin{equation}\label{degadd}(p+q)\in\mathcal{E}(d_1),\end{equation} 
except, of course, if $d_1=d_2$ and the lead coefficients of $p$ and $q$ are negatives of one another.

The next several lemmas will be needed in the proofs of  Theorems \ref{thm: main support} and \ref{thm: heights}. 
\begin{lem}\label{lem: zerosetmeasurezero}
For $i=1,2,3,4$ let $p_i\in\mathcal{E}_0(d_i)$ assume that $d_1> d_2,d_3,d_4.$  Then the set of points $\theta$ such that 
\begin{equation}\label{noteq}
|p_1(\theta)|^p+|p_2(\theta)|^p=|p_3(\theta)|^p+|p_4(\theta)|^p
\end{equation}
has measure zero.
\end{lem}

\begin{lem}\label{lem: 4psolutions}
Let $p\geq 1$ be an odd integer, and let $a,b,c,d,C\in\mathbb{R},$ $a,b,c,d\neq0.$ Let $p(\theta)= a+b\mathbbm{e}^{\mathbbm{i}\theta},$ and $q(\theta)=c+d\mathbbm{e}^{\mathbbm{i}\theta}.$ If there are more than $4p$ distinct $\theta\in[0,2\pi]$ such that
\begin{equation*}
|p(\theta)|^p-|q(\theta)|^p=C,
\end{equation*}
then $ab=cd$ and $a^2+b^2=c^2+d^2.$
 \end{lem}

\begin{lem}\label{lem: zerosetmeasurezero2}
Let $p\geq1 $ be an integer and let $a,b,c,d,C \in\mathbb{R}, \gamma>0, \kappa\neq 0,\pm 1.$ Then the set of $\theta$ such that
\begin{equation}\label{measurezero}
\left|a+b\mathbbm{e}^{\mathbbm{i}\theta}+c\mathbbm{e}^{\mathbbm{i}(\gamma+1)\theta}\right|^p-\left|\kappa a+\frac{1}{\kappa}b\mathbbm{e}^{\mathbbm{i}\theta}+\kappa c\mathbbm{e}^{\mathbbm{i}(\gamma+1)\theta}\right|^p=C
\end{equation}
has measure zero.
\end{lem}

\section{The proof of Theorem \ref{thm: main support}}
%In this section we will show that for a fixed central frequency $\xi$, and signal $x(t)$ of the form \textcolor{red}{add equation number here} that the function $f_\xi(s)=f_\xi[x](s)$ is a piecewise linear function with knots are given by the difference set  $\mathcal{D}(x)$ defined in \textcolor{red}{add}. We we also show that these measurements characterize the amplitudes $a_i$ with $\mathcal{O}(p)$ randomly chosen frequencies.

Before proving Theorem \ref{thm: main support} we will first prove a preliminary result which shows, even without the assumption that $x(t)$ is collision free, that $f_\xi(s)$ is a peicewise linear function whose set of knots is contained in $\mathcal{D}(x)$.
This result is based on the observation that we may write \begin{equation*}
    f_\xi(s)=\sum_{i<j} \alpha_{i,j}(s)\beta_{i,j}(\xi).
\end{equation*}
where for each $i<j$,
\begin{equation}\label{yform}
\beta_{i,j}(\xi) \coloneqq \sum_{\ell=i}^j a_\ell\mathbbm{e}^{\mathbbm{i}\xi \Delta_{i,\ell}}
\end{equation}
is a function that only depends on $\xi$ and $\alpha_{i,j}(s)$ is piecewise linear function of $s$ whose singularities are contained in $\mathcal{D}(x).$ 

Specifically, we prove the following theorem. We emphasize that this result does not assume that $x(t)$ is collision free, which is why for $d\in\mathcal{D}(x)$ there might be multiple $i,j$ such that $\Delta_{i,j}=d$.

\begin{thm} \label{thm: piecewiselinear}
Let $p\geq 1$ be an integer, and assume $w(t)=1_{[0,1]}(t).$ For $i\leq j,$ $\beta_{i,j}(\xi)$ be as in \eqref{yform}. 
Then, for every fixed $\xi,$ the function
 $f_\xi(s) =  \|g_{s,\xi}\ast x\|_p^p$ is piecewise linear,
and $\partial_s^2 f_\xi(s)$ is a grid-free sparse signal whose  support is contained in $\mathcal{D}(x).$ %the set of pairwise distances, 
Specifically,
\begin{equation}\label{eqn: 2nd deriv}
    \partial_s^2 f_\xi(s) = \sum_{d\in \mathcal{D}(x)} \left(\sum_{\Delta_{i,j}=d} c_{i,j}(\xi)\right) \delta_{d},
\end{equation}
where %$\Delta_{i,j}=v_j-v_i$.
%and
\begin{equation}\label{fdoubleprimespecial}
    c_{i,i+1} (\xi) = |\beta_{i,i+1}(\xi)|^p-|\beta_{i+1,i+1}(\xi)|^p-|\beta_{i,i}(\xi)|^p
\end{equation}
and 
for $j\geq i+2$
\begin{equation}\label{fdoubleprimedistinct}
    c_{i,j}(\xi)= |\beta_{i,j}(\xi)|^p+|\beta_{i+1,j-1}(\xi)|^p-|\beta_{i+1,j}(\xi)|^p-|\beta_{i,j-1}(\xi)|^p.
\end{equation}
\end{thm}

\begin{proof}

We first note that 
\begin{align*}
|(g_{s,\xi}\star x)(t)| & =\left|\sum_{i=1}^k a_ig_{s,\xi}(t-v_i)\right|\\
& =\left|\sum_{i=1}^k a_i\mathbbm{e}^{\mathbbm{i}\xi(t-v_i)}1_{[v_i,v_i+s]}(t)\right|\\
& =\left|\sum_{i=1}^k a_i\mathbbm{e}^{-\mathbbm{i}\xi v_i}1_{[v_i,v_i+s]}(t)\right|.
\end{align*}
For  $I\subseteq \{1,\ldots, k\},$ let $R_I(s)$ be the set of $t$ for which $a_i\mathbbm{e}^{-\mathbbm{i}\xi v_i}1_{[v_i,v_i+s]}(t)$ is nonzero if and only if $i\in I$, i.e., 
\begin{equation*}
R_I(s) = \{t: t\in [v_i,v_i+s]\; \forall i\in I, t\notin [v_i,v_i+s]\; \forall i\notin I\}.
\end{equation*}
  Then, since $w(t)=1_{[0,1]}(t)$ it is clear that for $t\in R_I,$ \begin{equation*}
 |(g_{s,\xi}\star x)(t)| = \left|\sum_{i\in I} a_i \mathbbm{e}^{-\mathbbm{i}\xi v_i}\right| \eqqcolon y_I(\xi).
\end{equation*} 
Therefore, 
\begin{equation}\label{fform} f_\xi(s) =\|(g_{s,\xi}\star x)(t)\|_p^p= \sum_{I\subseteq \{1,\ldots k\}} |y_I(\xi)|^p|R_I(s)|,
\end{equation}
where $|R_I(s)|$ denotes the Lebesgue measure of $R_I(s)$.
We will show that  for all $I\subseteq\{1,\ldots,k\}$,  $|R_I(s)|$ is  piecewise linear function  whose knots are contained in $\mathcal{D}(x).$

First, we note that $R_I(s)=\emptyset$ unless $I$ has the form $\{i,i+1,\ldots, j-1, j\}$ for some $i\leq j.$ Therefore,
\begin{equation}\label{eqn: modulation times set size} f_s(\xi) = \sum_{i=1}^k\sum_{j=i}^k |\beta_{i,j}(\xi)|^p|R_{i,j}(s)|,
\end{equation}
where  $R_{i,j}\coloneqq R_{\{i,\ldots,j\}}.$ 
and, as in \eqref{yform}, $\beta_{i,j}(\xi)$ is given by %the exponential polynomial defined (\ref{eqn: ydefinition}) and 
\begin{equation*}
|\beta_{i,j}(\xi)| = \left|\sum_{\ell=i}^j a_\ell\mathbbm{e}^{\mathbbm{i}\xi \Delta_{i,\ell}}\right| = \left|\sum_{\ell=i}^j a_\ell\mathbbm{e}^{\mathbbm{i}\xi v_\ell}\right|. %y_{\{i,\ldots,l\}}(\xi).
\end{equation*}
%In particular, $\beta_{i,j}$ is a fucntion of the pairwise distances, $v_j-v_i,$ and the $a_i.$ 

Now, turning our attention to $R_{i,j}(s),$ we observe by definition that a point $t$ is in $R_{i,j}(s)$ if and only if it satisfies the following three conditions:
\begin{align*}
v_\ell\leq &t \leq v_\ell+s \quad \text{for all } i\leq \ell \leq j,\\
&t>v_{i-1}+s, \text{ and }\\
&t<v_{j+1}.
\end{align*}
Therefore, letting   $\lor(a,b)$ and $\wedge(a,b)$  denote $\min\{a,b\}$ and $\max\{a,b\}$, we see
\begin{align}
\label{R} R_{i,j}(s) &= [v_j,v_i+s]\cap [v_{i-1}+s,v_{j+1}] \\
&= [v_j \lor (v_{i-1}+s), (v_i+s) \wedge v_{j+1}],
\end{align}
and therefore 
\begin{equation*}
|R_{i,j}(s)| = ((v_i+s) \wedge v_{j+1})-(v_j \lor (v_{i-1}+s))
\end{equation*}
if the above quantity is positive and zero otherwise. 
 It follows from $(\ref{R})$ that $|R_{i,j}(s)|$ is a piecewise linear function, and that $\partial_s^2|R_{i,j}(s)|$ is given by 
%%%%%%%%%%%%%%%%%%%%%%%%%%%%%%%%%%%%%%%%%%%%%%%%%%%%
%STOP HERE
%%%%%%%%%%%%%%%%%%%%%%%%%%%%%%%%%%%%%%%%%%%%%%%%%%

\begin{equation}\label{Rdirac}
\partial_s^2 |R_{i,j{(S)}}| = \delta_{\Delta_{i,j}}(s) + \delta_{\Delta_{i-1,j+1}}(s) - \delta_{\Delta_{i-1,j}}(s) - \delta_{\Delta_{i,j+1}}(s).
\end{equation}
We note that in order for this equation to be valid for all $1\leq i<j\leq k,$ we  identify $v_0$ and $v_{k+1}$ with $-\infty$ and $\infty,$ and therefore, $\delta_{{\Delta_{0,j}}}$ $\delta_{{\Delta_{i-1,k+1}}}$  are interpreted as being the zero function since the domain of $f$ is $(0,\infty).$ Likewise $\delta_{{\Delta_{i,i}}}=\delta_0$ is interpreted as the zero function in the above equation.

Combining (\ref{Rdirac}) with (\ref{eqn: modulation times set size}) implies that $\partial_s^2 f_\xi(s)$ is a sparse signal with support contained in $\mathcal{D}(x),$ and for $d\in\mathcal{D}(x),$

\begin{equation*}
    \partial_s^2 f_\xi(d)=\sum_{\Delta_{i,j}=d}c_{i,j}(\xi)
\end{equation*} as desired. 
\end{proof}
Before we prove Theorem \ref{thm: main support}, we note the following example which shows that, in general, the support of $\partial_s^2 f_\xi(s)$ may be a proper subset of $\mathcal{D}(x).$

\begin{example} \label{ex: collision example}
If $p=2$ and \begin{equation*}
x(t) = \delta_{1}(t)+\delta_{2}(t)+\delta_{3}(t)-\delta_{4}(t),
\end{equation*}
then $2\in\mathcal{D}(x),$ but 
\begin{equation*}\partial_s^2 f_\xi(2)=0.
\end{equation*}
%for every $\xi$ despite the fact that $2=\Delta_{1,3}=\Delta_{2,4}.$
\end{example}
\begin{proof}For this choice of $x,$
there are two pairs $(i,j)$ such that $\Delta_{i,j}=2,$ namely $(1,3)$ and $(2,4)$. Therefore, by Theorem \ref{thm: piecewiselinear},
\begin{align*}
  \partial_s^2 f_\xi(2) &= \left(|y_{1,3}(\xi)|^2 + |y_{2,2}(\xi)|^2  - |y_{1,2}(\xi)|^2 - |y_{2,3}(\xi)|^2\right) \\&\:\:\:+ \left(|y_{2,4}(\xi)|^2 + |y_{3,3}(\xi)|^2  - |y_{2,3}(\xi)|^2 - |y_{3,4}(\xi)|^2\right).
\end{align*}
Inserting $(a_1,a_2,a_3,a_4)=(1,1,1,-1),$ $\Delta_{i,i+1}=1,$ and $\Delta_{i,i+2}=2$ into (\ref{yform}) implies that
\begin{align*}
\partial_s^2 f_\xi(2) &= \left(|1+\mathbbm{e}^{\mathbbm{i}\xi}+\mathbbm{e}^{2i\xi}|^2 + 1  - |1+\mathbbm{e}^{\mathbbm{i}\xi}|^2 - |1+\mathbbm{e}^{\mathbbm{i}\xi}|^2\right) \\&\:\:\:+ \left(|1+\mathbbm{e}^{\mathbbm{i}\xi}-\mathbbm{e}^{2i\xi}|^2 + 1 - |1+\mathbbm{e}^{\mathbbm{i}\xi}|^2 - |1-\mathbbm{e}^{\mathbbm{i}\xi}|^2\right)\\
&=|1+\mathbbm{e}^{\mathbbm{i}\xi}+\mathbbm{e}^{2i\xi}|^2+|1+\mathbbm{e}^{\mathbbm{i}\xi}-\mathbbm{e}^{2i\xi}|^2 \\&\:\:\:+2 - 3|1+\mathbbm{e}^{\mathbbm{i}\xi}|^2-|1-\mathbbm{e}^{\mathbbm{i}\xi}|^2\\
&=0.
\end{align*}
The last inequality follows from repeatedly applying the the trigonometric identities $\sin^2(\theta)+\cos^2(\theta)=1$ and $\cos(\theta)=\cos(2\theta)\cos(\theta)+\sin(2\theta)\sin(\theta).$
\end{proof}

We shall now prove Theorem \ref{thm: main support}.

\begin{proof}[The Proof of Theorem \ref{thm: main support}]
By assumption,  $x(t)$ is collision free. Therefore, for all $d\in\mathcal{D}(x)$, there is a unique $i,j$ such that $\Delta_{i,j}=d$, and so, by \eqref{eqn: 2nd deriv},   it suffices  to show that $c_{i,j}(\xi)\neq 0$ for all $i<j$ and for almost every $\xi\in\mathbb{R},$ where as in \eqref{fdoubleprimespecial} and for $j\geq i+$ \eqref{fdoubleprimedistinct}
\begin{equation*}
  c_{i,i+1}(\xi) = |\beta_{i,i+1}(\xi)|^p  - |\beta_{i+1,i+1}(\xi)|^p - |\beta_{i,i}(\xi)|^p,
\end{equation*}
and for $j\geq i+2$,
\begin{equation*}
 c_{i,j}(\xi) =  |\beta_{i,j}(\xi)|^p + |\beta_{i+1,j-1}(\xi)|^p - |\beta_{i+1,j}(\xi)|^p - |\beta_{i,j-1}(\xi)|^p,
\end{equation*}
where
\begin{equation*}
\beta_{i,j}(\xi) = \left|\sum_{k=i}^ja_k \mathbbm{e}^{-\mathbbm{i}\xi\Delta_{i,k}}\right|.
\end{equation*}
Observe that $\beta_{i,j}(\xi)$ are generalized exponential Laurent polynomials of the form introduced in Section \ref{sec: exppoly}, and in particular, $\beta_{i,j}\in\mathcal{E}_0(\Delta_{i,j}).$ Therefore, when $j\geq i+2,$ it follows from Lemma \ref{lem: zerosetmeasurezero} that $c_{i,j}(\xi)$ vanishes on a set of measure zero since if  $c_{i,j}(\xi)=0$  we have
\begin{equation*}  |\beta_{i,j}(\xi)|^p + |\beta_{i+1,j-1}(\xi)|^p = |\beta_{i+1,j}(\xi)|^p + |\beta_{i,j-1}(\xi)|^p.
\end{equation*}

In the case where $j=i+1,$ we see that \begin{equation*}
c_{i,i+1}(\xi) = |a_i+a_{i+1}\mathbbm{e}^{-\mathbbm{i}\xi\Delta_{i,i+1}}|^p-|a_i|^p-|a_{i+1}|^p,
\end{equation*}
%and that therefore $c_{i,i+1}(\xi)$ is nonzero except at countably many $\xi$ since at any $\xi$ for which $\partial_s^2 f_\xi(\Delta_{i,i+1})=0,$ $\xi\Delta_{i,i+1}$ must be a zero of the trigonometric polynomial
For any $\xi$ such that $c_{i,i+1}(\xi)=0,$ we see that $\xi\Delta_{i,i+1}$ is a solution to
\begin{equation*}
\left|a_i + a_{i+1}\mathbbm{e}^{\mathbbm{i}\theta}\right|^2-\left(|a_i|^p+|a_{i+1}|^p\right)^{2/p}=0.\end{equation*}
Thus, $c_{i,i+1}(\xi)$ vanishes on a set of measure zero since the left-hand side of the above equation is a trigonometric polynomial.

\end{proof}

\section{The Proof of Theorems \ref{thm: heights}}
 \label{sec: heights proof}

\begin{proof}Let $\xi_1,\xi_2,\ldots,\xi_L$ be i.i.d. standard normal random variables. Since $x$ is collision free, with probability one, each of the $\xi_\ell\Delta_{i,i+1}(x)$ are distinct modulo $2\pi$, i.e.
\begin{equation}\label{distinctfrequencies}
\xi_\ell\Delta_{i,i+1}(x) \not \equiv \xi_{\ell'}\Delta_{i',i'+1}(x)\mod 2\pi
\end{equation}
for all $1\leq i,i'\leq k-1$ and $1\leq \ell,\ell'\leq L,$ except when $(i,\ell)= (i',\ell').$
For the rest of the proof we will assume this is the case. %$\vec{b}=(b_1,\ldots,b_k)$ be a vector such that there exists a signal 

Let \begin{equation*}
 \widetilde{x}(t)=\sum_{j=1}^k \widetilde{a}_j \delta_{\widetilde{v}_j}(t)
\end{equation*} be a signal  $\mathcal{D}(\widetilde{x})=\mathcal{D}(x),$ $\|\vec{a}\|_p=\|\vec{\widetilde{a}}\|_p$, and $\partial_s^2f_{\xi_\ell}[x](d)=\partial_s^2f_{\xi_\ell}[\widetilde{x}](d)$ for all $d\in\mathcal{D}(x)$  and for all $1\leq \ell \leq L-1.$  Note that $\widetilde{x}(t)$ depends on $\xi_1,\ldots, \xi_{L-1},$ but is independent of $\xi_L.$ By assumption that $x(t)$ and $\widetilde{x}(t)$ are collision free (and also, as discussed in the Section \ref{sec: intro}, we assume that we are not in the special case where $k=6$ and the $\{v_j\}_{j=1}^k$ belong to a special parametrized family). Therefore, the fact that $\mathcal{D}(x)=\mathcal{D}(\widetilde{x})$ implies that the support sets of $x$ and $\widetilde{x}$ are equivalent up to translation and reflection, so we may assume without loss of generality that $\Delta_{i,j}(x)=\Delta_{i,j}(\widetilde{x})\eqqcolon\Delta_{i,j}$ for all $1\leq i\leq j\leq k.$ 

We will show that $\vec{\widetilde{a}}$ must be given by
\begin{equation}\label{eqn: alternation}
\widetilde{a}_i  = \begin{cases} 
      \frac{1}{c}a_i &  \text{if } i \text{ is odd} \\
      ca_i & \text{if } i \text{ is even}
   \end{cases},
\end{equation}
where $c=\pm 1$ or
\begin{equation}\label{eqn: alternate c}
|c|^p= \frac{\sum_{i=1}^{\lfloor \frac{k+1}{2}\rfloor}|a_{2i-1}|^p}{\sum_{i=1}^{\lfloor \frac{k}{2}\rfloor}|a_{2i}|^p}.
\end{equation}
Then, we will show that,  if $c$ satisfies \eqref{eqn: alternate c}, but $c\neq \pm 1,$ 
then with probability one  $$\partial_s^2f_{\xi_L}[x](\Delta_{1,3})\neq \partial_s^2f_{\xi_L}[\widetilde{x}](\Delta_{1,3}).$$
Since $\widetilde{x}(t)$ \big{(}and therefore $\vec{\widetilde{a}}$\big{)} was chosen to depend on $\xi_1,\ldots,\xi_{L-1},$ but not $\xi_L,$ these two facts together will imply that, with probability one, if $\widetilde{x}(t)$ is any signal such that  $\mathcal{D}(\widetilde{x})=\mathcal{D}(x)$ and $\partial_s^2f_{\xi_\ell}[x](d)=\partial_s^2f_{\xi_\ell}[\widetilde{x}](d)$ for all $d\in\mathcal{D}$  and all $1\leq \ell \leq L,$ then $\vec{\widetilde{a}}=\pm \vec{a}$ and therefore $\widetilde{x}(t)$ is equivalent to  $\pm x(t)$ up to reflection and translation.

We first will show that \eqref{eqn: alternation} holds in the case where $p$ is odd.
 Setting $\partial_s^2f_{\xi_\ell}[x](\Delta_{i,i+1})
     =\partial_s^2f_{\xi_\ell}[\widetilde{x}](\Delta_{i,i+1})$ and using \eqref{fdoubleprimespecial} implies that for all $1\leq \ell\leq L-1$ and all $1\leq i \leq k-1$
we have 
\begin{align}
     &|a_i+a_{i+1}\mathbbm{e}^{\mathbbm{i}\xi_\ell\Delta_{i,i+1}}|^p-|a_{i+1}|^p-|a_{i}|^p\nonumber\\
     =&|\widetilde{a}_i+\widetilde{a}_{i+1}\mathbbm{e}^{\mathbbm{i}\xi_\ell\Delta_{i,i+1}}|^p-|\widetilde{a}_{i+1}|^p-|\widetilde{a}_{i}|^p.\label{eqn: first deltas same}
\end{align}  Therefore,  $\xi_1\Delta_{i,i+1},\ldots,\xi_{L-1}\Delta_{i,i+1}$ constitute  $L-1$ solutions, which are distinct modulo $2\pi,$ to the equation
\begin{equation*}
|a_i + a_{i+1}\mathbbm{e}^{\mathbbm{i}\theta}|^p-|\widetilde{a}_i+\widetilde{a}_{i+1}^{\mathbbm{i}\theta}|^p=|\widetilde{a}_i|^p+|\widetilde{a}_{i+1}|^p-|a_i|^p-|a_{i+1}|^p.
\end{equation*}

 Since $L-1> 4p,$ Lemma \ref{lem: 4psolutions} implies 
that
%We will show that the fact that $\vec{a}$ and $\vec{b}$ both satisfy (\ref{Aeq}) for all $1\leq i\leq k,$ $1\leq j\leq L-1$ implies that 
\begin{equation}\label{prodsame}
a_ia_{i+1}=\widetilde{a}_i\widetilde{a}_{i+1}
\quad\text{and}\quad 
%a_i^2+a_{i+1}^2= 
\widetilde{a}_i^2+\widetilde{a}_{i+1}^2
\end{equation}
for all $1\leq i \leq k-1.$
It follows from (\ref{prodsame}) that \eqref{eqn: alternation} holds  
%\begin{equation}\label{cform}
%\widetilde{a}_i=  \begin{cases} 
%      \frac{1}{c}a_i &  \text{if } i \text{ is odd} \\
%      ca_i & \text{if } i \text{ is even},
%   \end{cases}
%\end{equation}
with $c\coloneqq a_1/\widetilde{a}_1.$

Now consider the case  where $p=2m$ is even. 
Similarly to \eqref{eqn: first deltas same}, the assumption that  $\partial_s^2f_{\xi_\ell}[x](\Delta_{i,i+1})= \partial_s^2f_{\xi_\ell}[\widetilde{x}](\Delta_{i,i+1})$ implies that 
for all $1\leq \ell\leq L,$ $1\leq i \leq k-1,$
\begin{align*}
    &|a_{i}+a_{i+1}\mathbbm{e}^{\mathbbm{i}\xi_\ell\Delta_{i,i+1}}|^{2m}-|a_i|^{2m}-|a_{i+1}|^{2m}\\=&    |\widetilde{a}_{i}+\widetilde{a}_{i+1}\mathbbm{e}^{\mathbbm{i}\xi_\ell\Delta_{i,i+1}}|^{2m}-|\widetilde{a}_i|^{2m}-|\widetilde{a}_{i+1}|^{2m}.
\end{align*}
 Therefore, for all $1\leq i \leq k-1,$
$\xi_1\Delta_{i,i+1},\ldots,\xi_{L-1}\Delta_{i,i+1}$ are $L-1$ zeros of
\begin{align*}
h_i(\theta)&\coloneqq|a_{i}+a_{i+1}\mathbbm{e}^{\mathbbm{i}\theta}|^{2m}-|\widetilde{a}_{i}+\widetilde{a}_{i+1}\mathbbm{e}^{\mathbbm{i}\theta}|^{2m}\\
&\quad\quad+|\widetilde{a}_i|^{2m}+|\widetilde{a}_{i+1}|^{2m}-|a_i|^{2m}-|a_{i+1}|^{2m}
\end{align*}
which are distinct modulo $2\pi.$
Using the fact that 
\begin{equation*}
|a_{i}+a_{i+1}\mathbbm{e}^{\mathbbm{i}\theta}|^{2}=a_i^2+a_{i+1}^2+2a_ia_{i+1}\cos(\theta) %5\quad\text{and}\quad \widetilde{a}_{i}+\widetilde{a}_{i+1}\mathbbm{e}^{\mathbbm{i}\theta}|^{2}=\widetilde{a}_i^2+\widetilde{a}_{i+1}^2+2\widetilde{a}_i\widetilde{a}_{i+1}\cos(\theta)
\end{equation*}
one may verify that $h_i(\theta)$ is a trigonometric polynomial of degree at most $m$ given by 
\begin{align*}
h_i(\theta)&=(a_i^2+a_{i+1}^2+2a_ia_{i+1}\cos(\theta))^m\\&\quad-(\widetilde{a}_i^2+\widetilde{a}_{i+1}^2+2\widetilde{a}_i\widetilde{a}_{i+1}\cos(\theta))^m\\&\quad\quad+\widetilde{a}_i^{2m}+\widetilde{a}_{i+1}^{2m}-a_i^{2m}-a_{i+1}^{2m}
\end{align*}
 Thus, since $L-1> p=2m,$ this implies that $h(\theta)$ must be uniformly zero. In particular, setting the lead coefficient equal to zero implies 
\begin{equation*}
(a_ia_{i+1})^m=(\widetilde{a}_i\widetilde{a}_{i+1})^m
\end{equation*}
for all $1\leq i \leq k-1.$
Using the binomial theorem and setting the $\cos(\theta)$ coefficient equal to zero gives
\begin{equation*}
(a_i^2+a_{i+1}^2)^{m-1}a_ia_{i+1}=(\widetilde{a}_i^2+\widetilde{a}_{i+1}^2)^{m-1}\widetilde{a}_i\widetilde{a}_{i+1}.
\end{equation*}
Together, the last two equations imply
\begin{equation*}
a_i^2+a_{i+1}^2=\widetilde{a}_i^2+\widetilde{a}_{i+1}^2\quad\text{and}\quad a_ia_{i+1}=\widetilde{a}_i\widetilde{a}_{i+1}.
\end{equation*}
As in the case where $p$ was odd, this implies that \eqref{eqn: alternation} must  hold.

Combining (\ref{eqn: alternation}) with the assumption that 
\begin{equation*}
\sum_{i=1}^k |a_i|^p = \sum_{i=1}^k |\widetilde{a}_i|^p 
\end{equation*}
implies that either $c=\pm 1$ or that $c$ satisfies $\eqref{eqn: alternate c}.$ 
%\begin{equation*}
%|c|^p=\pm \frac{\sum_{i=1}^{\lfloor \frac{k+1}{2}\rfloor}|a_{2i-1}|^p}{\sum_{i=1}^{\lfloor \frac{k}{2}\rfloor}|a_{2i}|^p}.
%\end{equation*}
%Therefore, the proof will be complete once we prove (\ref{prodsame}) and (\ref{squaresum})  and also show that for any $\vec{b}$ of the form (\ref{cform}), the probability that $\vec{b}$ satisfies (\ref{Beq}) when $i=1,j=L$ is zero except for in the case that $c=\pm 1.$
%In particular, there are at most two possible $\vec{b},$ other than $\vec{b}=\pm\vec{a},$ which satisfy (\ref{Aeq}) for all $1\leq i \leq k-1.$ 
Thus the proof will be complete once we show that if $c$ satisfies \eqref{eqn: alternate c}, but $c\neq \pm1,$, then with probability one,  $\partial_s^2f_{\xi_L}[x](\Delta_{1,3})\neq \partial_s^2f_{\xi_L}[\widetilde{x}](\Delta_{1,3})$. %(Recall that $\vec{b}$ was chosen to depend on $\xi_1,\ldots,\xi_{L-1};$ however, $\xi_L$ is independent of $\vec{b}.$)

%Fix $1\leq i \leq k-1.$ Then since $\vec{a}$ and $\vec{b}$ both satisfy (\ref{Aeq}), it follows that  $\xi_1\Delta_{i,i+1},\ldots,\xi_{L-1}\Delta_{i,i+1}$ are $L-1$ solutions to \begin{equation*}
%|a_i + a_{i+1}\mathbbm{e}^{\mathbbm{i}\theta}|^p-|b_i+b_{i+1}^{\mathbbm{i}\theta}|^p=|b_i|^p+|b_{i+1}|^p-|a_i|^p-|a_{i+1}|^p
%\end{equation*}
%which are distinct modulo $2\pi.$ 
%Therefore, both  (\ref{prodsame}) and (\ref{squaresum}) must hold by Lemma \ref{lem: 4psolutions}.

By \eqref{fdoubleprimedistinct}, if $\partial_s^2f_{\xi_L}[x](\Delta_{1,3})= \partial_s^2f_{\xi_L}[\widetilde{x}](\Delta_{1,3}),$ %$\vec{a}$ also satisfies \eqref{Beq} when $i=1, j=L.$ Therefore, % the theorem will be proved once we show that if $\vec{b}$ is constructed as in (\ref{cform}) and satisfies (\ref{prodsame}) and (\ref{squaresum}), then $\vec{b}$  satisfies $(\ref{Beq}),$ when $j=L, i=1$ with probability zero except in the case that $c=\pm 1.$ 
%Since $\vec{a}$ satisfies $(\ref{Beq}),$ $\vec{b}$ also satisfying $(\ref{Beq})$  implies that
then
\begin{align}\label{unsimpleBeq}
&|a_1+a_{2}\mathbbm{e}^{\mathbbm{i}\xi_L\Delta_{1,2}}+a_{3}\mathbbm{e}^{\mathbbm{i}\xi_L\Delta_{1,3}}|^p+|a_{2}|^p\\
&\quad-|a_{2}\mathbbm{e}^{\mathbbm{i}\xi_L\Delta_{1,2}}+a_{3}\mathbbm{e}^{\mathbbm{i}\xi_L\Delta_{1,3}}|^p-|a_1+a_{2}\mathbbm{e}^{\mathbbm{i}\xi_L\Delta_{1,2}}|^p\nonumber\\
=&|\widetilde{a}_1+\widetilde{a}_{2}\mathbbm{e}^{\mathbbm{i}\xi_L\Delta_{1,2}}+\widetilde{a}_{3}\mathbbm{e}^{\mathbbm{i}\xi_L\Delta_{1,3}}|^p+|\widetilde{a}_{2}|^p\nonumber\\&\quad-|\widetilde{a}_{2}\mathbbm{e}^{\mathbbm{i}\xi_L\Delta_{1,2}}+\widetilde{a}_{3}\mathbbm{e}^{\mathbbm{i}\xi_L\Delta_{1,3}}|^p-|\widetilde{a}_1+\widetilde{a}_{2}\mathbbm{e}^{\mathbbm{i}\xi_L\Delta_{1,2}}|^p.\nonumber
\end{align}
But (\ref{prodsame}) implies that for all $i$ either $(a_i,a_{i+1})=\pm(\widetilde{a}_i,\widetilde{a}_{i+1})$ or $(a_i,a_{i+1})=\pm(\widetilde{a}_{i+1},\widetilde{a}_i).$ In either case, we have that  
\begin{equation*}
|a_1+a_{2}\mathbbm{e}^{\mathbbm{i}\xi_L\Delta_{1,2}}|=|\widetilde{a}_1+\widetilde{a}_{2}\mathbbm{e}^{\mathbbm{i}\xi_L\Delta_{1,2}}|\end{equation*}and\begin{equation*} |a_{2}\mathbbm{e}^{\mathbbm{i}\xi_L\Delta_{1,2}}+a_{3}\mathbbm{e}^{\mathbbm{i}\xi_L\Delta_{1,3}}|=|\widetilde{a}_{2}\mathbbm{e}^{\mathbbm{i}\xi_L\Delta_{1,2}}+\widetilde{a}_{3}\mathbbm{e}^{\mathbbm{i}\xi_L\Delta_{1,3}}|.
\end{equation*}
Combining this with (\ref{unsimpleBeq}) gives
\begin{align}
&|a_1+a_{2}\mathbbm{e}^{\mathbbm{i}\xi_L\Delta_{1,2}}+a_{3}\mathbbm{e}^{\mathbbm{i}\xi_L\Delta_{1,3}}|^p+|a_{2}|^p\nonumber\\
&=|\widetilde{a}_1+\widetilde{a}_{2}\mathbbm{e}^{\mathbbm{i}\xi_L\Delta_{1,2}}+\widetilde{a}_{3}\mathbbm{e}^{\mathbbm{i}\xi_L\Delta_{1,3}}|^p+|\widetilde{a}_{2}|^p\label{simpleBeq}.
\end{align}
However, by Lemma \ref{lem: zerosetmeasurezero2} the set of $\xi_L\in\mathbb{R}$ such that (\ref{simpleBeq}) holds has measure zero, unless $c=\pm 1$. Since $\xi_L$ is a normal random variable, this happens with probability  zero. Thus, the proof is complete.
\end{proof}

\section{Proofs of auxiliary lemmas}\label{sec: proof aux}
In this section, we will provide proofs of the lemmas stated in section \ref{sec: exppoly}.

\begin{proof}[The Proof of Lemma \ref{lem: welldefE}]
%We need to show that $p-q$ is not uniformly zero. Since $p-q\in\mathcal{E}$ 
By linearity, it suffices to show that $\alpha_1,\ldots,\alpha_N$ are nonzero numbers, then $p(\theta) = \sum_{k=1}^N \alpha_k\mathbbm{e}^{\mathbbm{i}\gamma_k\theta}$ is not the zero  function.
We will restrict attention to  the case where $|\gamma_N|>|\gamma_k|$ for all $1\leq k\leq N-1.$ The proofs of the other cases, where either $|\gamma_1|>|\gamma_k|$ for all $2\leq k \leq N$ or where $|\gamma_1|=|\gamma_N|>|\gamma_k|$ for all $2\leq k\leq N-1$ are similar.

For all $n\geq 1,$ the $n$-th derivative of $p(\theta)$ is given by
\begin{equation*}
p^{(n)}(\theta) = \sum_{k=1}^N \alpha_k\gamma_k^n\mathbbm{e}^{\mathbbm{i}\gamma_k\theta}.
\end{equation*}
Therefore, since $|\gamma_N|>\gamma_k$ for all $1\leq k \leq N-1,$ we have 
\begin{equation*}
\lim_{n\rightarrow\infty}\frac{p^{(n)}(0)}{\gamma_N^n}=\alpha_N.
\end{equation*}
In particular, there exists $n$ such that $p^{(n)}(0)\neq 0,$ and therefore $p(\theta)$ is not the zero function.
\end{proof}

\begin{proof}[The proof of Lemma \ref{lem: zerosetmeasurezero}]
In the case where $p=2m$ is even, then  by \eqref{square} and 
\eqref{degadd},   %$|p_i|^{2m}=\left(|p_i|^2\right)^m$ and so 
$|p_i(\theta)|^{2m}\in\mathcal{E}(md_i)$ for each $i.$ Therefore, since $d_1>d_2,d_3,d_4,$ it follows from (\ref{degadd}) that 
\begin{equation*}
|p_1(\theta)|^{2m}+|p_2(\theta)|^{2m}-|p_3(\theta)|^{2m}-|p_4(\theta)|^{2m}
\end{equation*}
is an element of $\mathcal{E}(md_1)$ and therefore vanishes on a set of measure zero. 

Now consider the case where $p=2m+1$ is odd.
 Squaring both sides of  (\ref{noteq}) implies
\begin{equation*}
p_5(\theta)=2(|p_1(\theta)p_2(\theta)|^{2m+1}-|p_3(\theta)p_4(\theta)|^{2m+1}),
\end{equation*}
where $$p_5(\theta)\coloneqq|p_1(\theta)|^{4m+2}+|p_2(\theta)|^{4m+2}-|p_3(\theta)|^{4m+2}-|p_4(\theta)|^{4m+2}.$$ Thus, 
squaring both sides again gives
\begin{equation*}
p_6(\theta)=8|p_1(\theta)p_2(\theta)p_3(\theta)p_4(\theta)|^{2m+1},
\end{equation*}
where $$p_6(\theta)\coloneqq p_5(\theta)^2-4|p_1(\theta)p_2(\theta)|^{4m+2}-4|p_3(\theta)p_4(\theta)|^{4m+2}.$$
Therefore, squaring both sides one final time implies that 
\begin{equation*}
p_6(\theta)^2-64|p_1(\theta)p_2(\theta)p_3(\theta)p_4(\theta)|^{4m+2}=0.
\end{equation*}
However, since $d_1>d_2,d_3,d_4,$  applying (\ref{degtimes}), (\ref{square}), and (\ref{degadd}), implies that $(p_6(\theta)^2-64|p_1(\theta)p_2(\theta)p_3(\theta)p_4(\theta)|^2)\in\mathcal{E}((8m+4)d_1)$ and therefore vanishes on a set of measure zero.
\end{proof}

\begin{proof}[The Proof of Lemma \ref{lem: 4psolutions}]
If $\theta$ is a solution to \begin{equation*}
|p(\theta)|^p-|q(\theta)|^p=C,
\end{equation*}
then
\begin{equation*}
|p(\theta)|^{2p}-|q(\theta)|^{2p}-C^2=2|q(\theta)|^pC.
\end{equation*}
Therefore, 
\begin{equation*}
f(\theta)\coloneqq\left(|p(\theta)|^{2p}-|q(\theta)|^{2p}-C^2\right)^2-4|q(\theta)|^{2p}C^2=0.
\end{equation*}
Since \begin{equation*}
|p(\theta)|^2=a^2+b^2+2ab\cos(\theta)\quad\text{and}\quad |q(\theta)|^2=c^2+d^2+cd\cos(\theta),
\end{equation*} $f(\theta)$ is a trigonometric polynomial of degree at most $2p$. Thus, if $f(\theta)$ has more than $4p$ zeros in $[0,2\pi]$ it must be uniformly zero. Expanding out terms, we see 
\begin{align*}
&f(\theta)\\
=& ((a^2+b^2+2ab\cos(\theta))^p - (c^2+d^2+2cd\cos(\theta))^p-C^2)^2\\&-4C^2(c^2+d^2+2cd\cos(\theta))^p
\end{align*}
and so setting the $\cos^{2p}(\theta)$ coefficient equal to zero implies
\begin{equation*}
0=(2^pa^pb^p-2^pc^pd^p)^2
\end{equation*}
which implies $ab=cd$ since $p$ is odd. If $p\geq 3,$ then we have $2(p-1)>p.$ Therefore,
\begin{equation*}
    f_1(\theta)\coloneqq((a^2+b^2+2ab\cos(\theta))^p - (c^2+d^2+2cd\cos(\theta))^p-C^2)^2
\end{equation*}
has strictly greater degree than 
\begin{equation*}
f_2(\theta)\coloneqq 4C^2(c^2+d^2+2cd\cos(\theta))^p
\end{equation*}
and so $f_1(\theta)$ must also be uniformly zero. 
Setting the $\cos^{p-1}(\theta)$ coefficient of $f_1(\theta)^{1/2}$ equal to zero yields
\begin{equation*}
\left(p(a^2+b^2)(2ab)^{p-1}- p(c^2+d^2)(2cd)^{p-1}\right)^2=0,
\end{equation*}
but since $ab=cd$ this implies that $a^2+b^2=c^2+d^2.$
On the other hand, if $p=1,$ using the fact that $ab=cd$ we see that 
\begin{equation*}
0=\left(a^2+b^2-(c^2+d^2)-C\right)^2- 4C^2(c^2+d^2+2cd\cos(\theta)),
\end{equation*}
which can only happen for all $\theta$ if $C=0$ and $a^2+b^2=c^2+d^2.$
\end{proof}

\begin{proof}[The Proof of Lemma \ref{lem: zerosetmeasurezero2}]
Let
\begin{equation*}
p(\theta)= a+b\mathbbm{e}^{\mathbbm{i}\theta}+c\mathbbm{e}^{\mathbbm{i}(\gamma+1)\theta}\end{equation*} and \begin{equation*} q(\theta)=\kappa a+\frac{1}{\kappa}b\mathbbm{e}^{\mathbbm{i}\theta}+\kappa c\mathbbm{e}^{\mathbbm{i}(\gamma+1)\theta}.
\end{equation*}
%Then by  (\ref{measurezero}) we see 
%\begin{equation*}
%|p(\theta)|^p=|q(\theta)|^p+C,
%\end{equation*}
Then squaring both sides of \eqref{measurezero} yields,
\begin{equation*}
|p(\theta)|^{2p}-|q(\theta)|^{2p}-C^2=2|q(\theta)|^pC,
\end{equation*}
and therefore if $\theta$ satisfies (\ref{measurezero}) it is a solution to $f(\theta)=0,$ where 
\begin{equation*}
f(\theta) \coloneqq \left(|p(\theta)|^{2p}-|q(\theta)|^{2p}-C^2\right)^2-4|q(\theta)|^{2p}C^2.
\end{equation*}
$f(\theta)$ is an element of the class $\mathcal{E}$ of generalized exponential polynomials introduced earlier. Thus, it will follow that $f$  vanishes on a set of measure zero as soon as we show that  $f$ is not uniformly zero. We will verify that  that the lead cofficient of $f$ is nonzero unless $\kappa=\pm1$.
Using the identity $\mathbbm{e}^{\mathbbm{i}x}+\mathbbm{e}^{-\mathbbm{i}x}=2\cos(x)$ we see that
\begin{align*}
|p(\theta)|^2 &= a^2+b^2+c^2+2ab\cos(\theta)\\&\quad\quad+2bc\cos(\gamma\theta)+2ac\cos((\gamma+1)\theta)
\end{align*}
and likewise  is given by 
\begin{align*}
 |q(\theta)|^2&=\kappa^2a^2+\frac{1}{\kappa^2}b^2+\kappa^2c^2\\&\quad\quad +2ab\cos(\theta)+2bc\cos(\gamma\theta)+2\kappa^2ac\cos((\gamma+1)\theta).
\end{align*}
Therefore, the lead coefficient of $f(\theta)$ vanishes if and only if $\kappa^2=1.$

\end{proof}

\section{Signal Synthesis}\label{sec: numerics}
In order to illustrate the ability of the measurements \ref{def: gabor measurements} to characterize sparse signals, we verify empirically that signals with the same measurement differ only by global reflections, translations and sign changes. Specifically, given a signal $x(t)=\sum_{j=1}^k a_j \delta_{v_j}(t)$ and a finite collection of measurements 
\begin{equation*}
    \{\|g_{s_k,\xi_k}\star x\|_p\}_{k=1}^N
\end{equation*}
we use a greedy  scheme to find another signal $\widetilde{x}(t)=\sum_{j=1}^k \widetilde{a}_j \delta_{v_j}(t)$ which minimizes
\begin{equation*}
    \sum_{k=1}^N\left|\|g_{s_k,\xi_k}\star x\|_p-\|g_{s_k,\xi_k}\star \widetilde{x}\|_p\right|^2.
\end{equation*}
Figure \ref{fig:diracs} shows the result of several signals $\tilde{x}$ which were obtained by solving this minimization problem.

\begin{figure}
    \centering
    \subfloat{\label{fig:syn_directed}\includegraphics[scale=0.15]{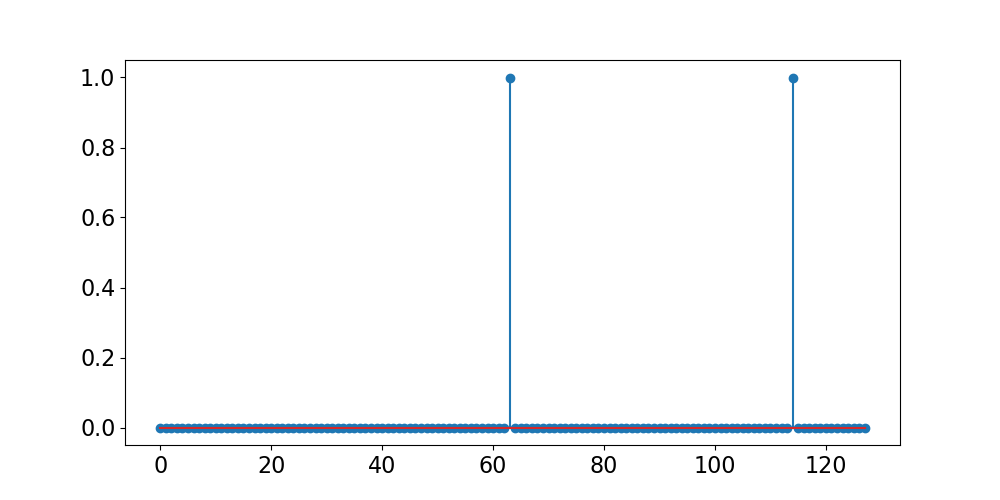}}
    \subfloat{\label{fig:syn_sym}\includegraphics[scale=0.15]{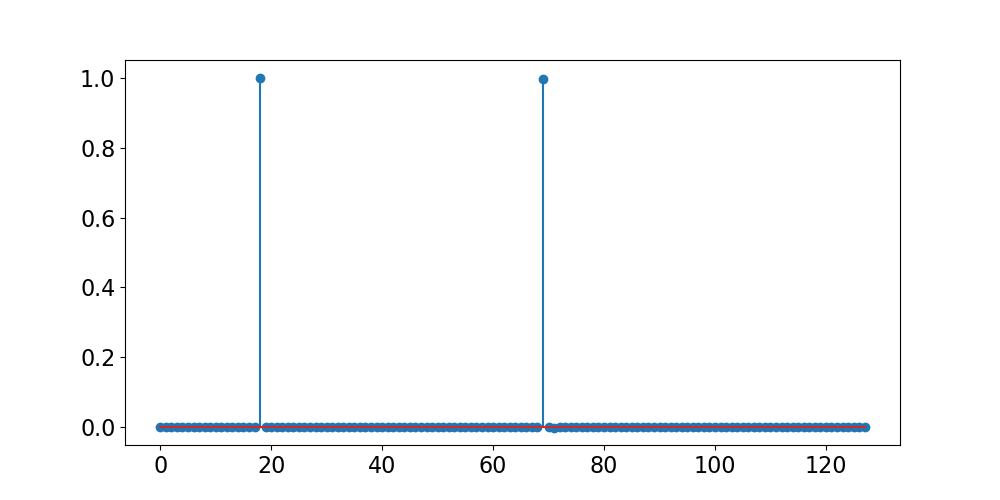}}\newline
\subfloat{\label{fig:syn_directed}\includegraphics[scale=0.15]{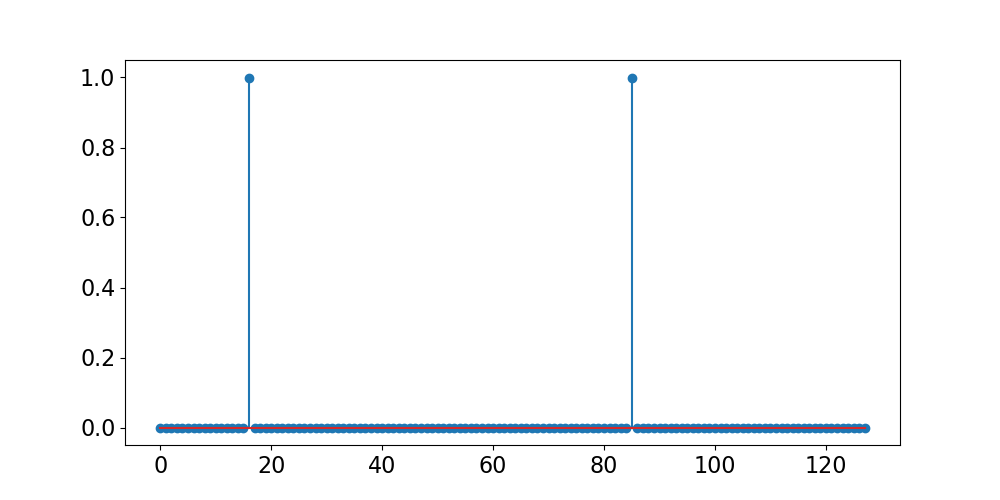}}
    \subfloat{\label{fig:syn_sym}\includegraphics[scale=0.15]{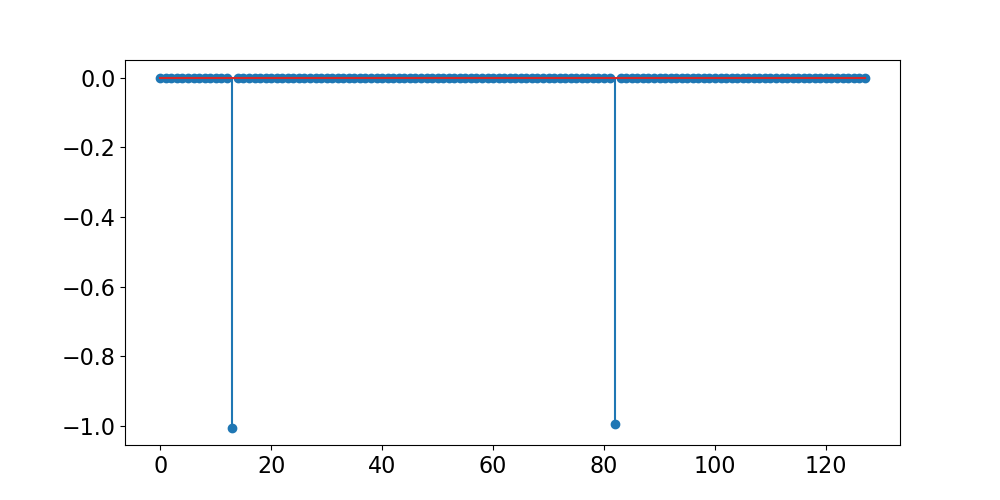}}\newline
    \subfloat{\label{fig:syn_directed}\includegraphics[scale=0.15]{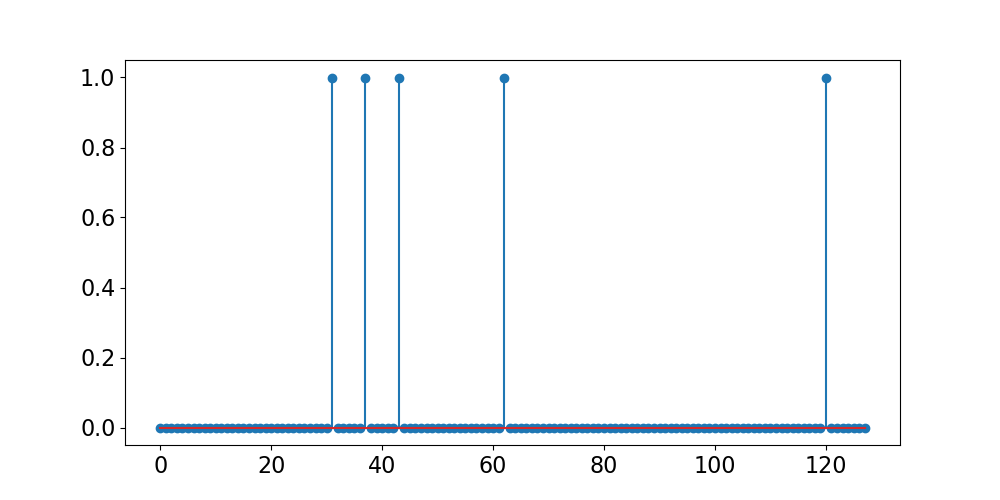}}
    \subfloat{\label{fig:syn_sym}\includegraphics[scale=0.15]{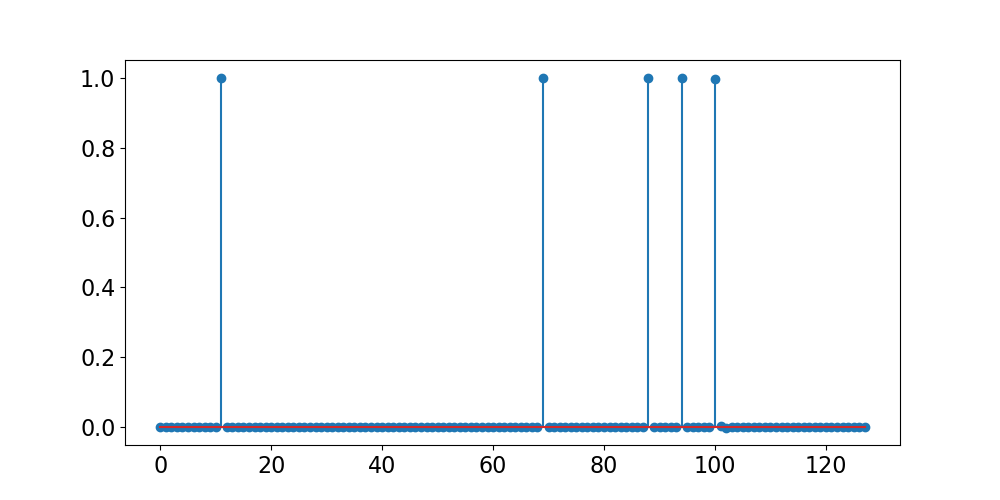}}
    \caption{Sparse signals reconstructed up to a global reflections, translations, and sign changes. Originals signals are on the left and reconstructed signals are on the right.}
    \label{fig:diracs}
\end{figure}
 
In all of experiments, we set the signal lengther to be $N=128$ and used two frequencies $\xi=(41/N)\pi,(23/N)\pi$. For the first signal we used scales $s=1, 14, 27, 40, 53, 2, 65, 96, 106$. For the second singal we used $s=1, 6, 11, 16, 21, 26, 31, 36, 41, 46, 51, 56, 61, 2, 4, 65, 96, 106$, and for the third we used  $s=1, 7, 13, 19, 25, 31, 37, 43, 49, 55, 61, 2, 4, 6, 65, 96, 106.$
\bibliographystyle{plain}

\end{document}